\documentclass[11pt]{article}
\usepackage{amsfonts}
\usepackage{amsmath,amssymb}
\usepackage{hyperref}
\hypersetup{colorlinks=true}
\usepackage{cleveref}
\usepackage{algorithm,algorithmicx,algpseudocode,graphicx}
\usepackage{amsthm,amsmath}
\usepackage{thm-restate}
\usepackage{color}

\usepackage{authblk}
\usepackage{fullpage}
\usepackage{comment}
\usepackage{mathtools}
\usepackage{authblk}

\usepackage{tikz}
\usepackage{tikz-cd}
\usetikzlibrary{calc,shapes,arrows.meta,matrix,decorations.pathreplacing,angles,quotes}

\newtheorem{definition}{Definition}

\newtheorem{lemma}{Lemma}
\newtheorem{claim}{Claim}
\newtheorem{theorem}{Theorem}
\newtheorem{hypothesis}{Hypothesis}
\newtheorem{openproblem}{Open Problem}

\newcommand{\poly}{\mathrm{poly}}

\newcommand{\opt}{\mathrm{opt}}
\newcommand{\cost}{\mathrm{cost}}
\newcommand{\tspan}{\mathrm{span}}
\newcommand{\bigmid}{\bigm|}
\newcommand{\Bigmid}{\Bigm|}
\newcommand{\ceil}[1]{\lceil #1 \rceil}

\newcommand{\NP}{\text{NP}}
\newcommand{\W}[1]{\text{W[#1]}}
\newcommand{\FPT}{\text{FPT}}
\newcommand{\ETH}{\text{ETH}}

\newcommand{\prf}{\Pi} 
\newcommand{\indexset}{\mathbb{F}_q^k}
\newcommand{\columnset}{\mathbb{F}_q^d}
\newcommand{\dist}{\mathrm{Hamming}}
\newcommand{\diff}{\mathrm{diff}}
\newcommand{\vs}{\mathrm{vertexset}}

\newcommand{\rand}[1]{\marginpar{\raggedright\footnotesize #1}}

\def\DEBUG{1}

\ifx\DEBUG\defined

  \newcommand{\yijia}[1]{\rand{\textbf{Yijia: }#1}}
  \newcommand{\yi}[1]{{\rand{\textbf{Yi: }#1}}}
  
  \newcommand{\yanlin}[1]{{\rand{\textbf{Yanlin: }#1}}}
\else
  
  \newcommand{\yijia}[1]{}
  \newcommand{\yi}[1]{}
  
  \newcommand{\yanlin}[1]{}
\fi

\title{Simple Combinatorial Construction of the $k^{o(1)}$-Lower Bound
for Approximating the Parameterized $k$-Clique}
\author[1]{Yijia Chen}
\affil[1]{Shanghai Jiao Tong University}
\author[2]{Yi Feng}
\affil[2]{Shanghai University of Finance and Economics}
\author[3]{Bundit Laekhanukit}
\affil[3]{Independent Researcher}
\author[4]{Yanlin Liu}
\affil[4]{Ocean University of China}

\begin{document}
    
    \date{}
    \maketitle
  
    \begin{abstract}
        In the parameterized $k$-clique problem, or $k$-Clique for short, we are given a graph $G$ and a parameter $k\ge 1$. The goal is to decide whether there exist $k$ vertices in $G$ that induce a complete subgraph (i.e., a $k$-clique). This problem plays a central role in the theory of parameterized intractability as one of the first \W 1-complete problems. Existing research has shown that even an \FPT-approximation algorithm for $k$-Clique with arbitrary ratio does not exist, assuming the Gap-Exponential-Time Hypothesis (Gap-\ETH)~[Chalermsook~et~al., FOCS'17 and SICOMP]. However, whether this inapproximability result can be based on the standard assumption of $\W 1\ne \FPT$ remains unclear.
        
        In a breakthrough work of Bingkai Lin~[STOC'21], any constant-factor approximation of $k$-Clique is shown to be \W 1-hard, and subsequently, the inapproximation ratio is improved to $k^{o(1)}$ in the work of Karthik~C.S.\ and Khot [CCC'22], and independently in [Lin, Ren, Sun Wang;
        ICALP'22] (under the apparently stronger complexity assumption ETH).
        All the work along this line follows the framework developed by Lin, which starts from the $k$-vector-sum problem and requires some involved algebraic techniques.

        This paper presents an alternative framework for proving the \W 1-hardness of the $k^{o(1)}$-\FPT-approximation of $k$-Clique. Using this framework, we obtain a gap-producing self-reduction of $k$-Clique without any intermediate algebraic problem. More precisely, we reduce from $(k,k-1)$-Gap Clique to $(q^k, q^{k-1})$-Gap Clique, for any function $q$ depending only on the parameter $k$, thus implying the $k^{o(1)}$-inapproximability result when $q$ is sufficiently large. Our proof is relatively simple and mostly combinatorial. At the core of our construction is a novel encoding of $k$-element subset stemmed from the theory of {\em network coding} and a {\em (linear) Sidon set} representation of a graph.
  
    \end{abstract}

    \section{Introduction}
    \label{sec:intro}

    In the {\em $k$-Clique} problem, we are given an $n$-vertex graph $G$ and an integer $k\ge 1$. The goal is to determine whether $G$ has a clique of size $k$.
    This problem has been a major point of interest in computer science and is one of Karp's 21 NP-complete problems~\cite{Karp72}.
    The computational intractability of $k$-Clique has also been studied in the context of approximation algorithms.
    In the wake of the celebrated PCP theorem, a series of works~\cite{FeigeGLSS96,Hastad01} have been dedicated to the construction of PCPs that capture the inapproximability of $k$-Clique (or more precisely, the {\em Maximum Clique} problem).
    To date, we know that even finding a clique of size $n^{\epsilon}$ from a graph $G$ promised to have a clique of size $n^{1-\epsilon}$, for any constant $\epsilon>0$, is \NP-hard.
    The intractability of $k$-Clique appears again in the context of Parameterized Complexity. Here, we parameterize the $k$-Clique problem by the number $k$ with the expectation that $k$ is much smaller than the size of $G$. Abusing the notation, we call the resulting parameterized problem again $k$-Clique. It is well known that $k$-Clique is a $\W 1$-complete problem, hence admits no \FPT-algorithms, i.e., algorithms with running time $t(k)\poly(n)$ for a computable function $t$, unless $\FPT= \W 1$. Consequently, there is no \FPT-algorithm that, on an input graph $G$ promised to have a $k$-clique, computes a clique of size $k$.
    A natural question arises whether there is an \FPT-algorithm that computes a clique of size $k/f(k)$. Here, $f$ is a computable function such that $k/f(k)$ is unbounded and non-decreasing. Such an algorithm is called an \FPT-approximation of $k$-Clique with approximation ratio $f$.
    Under an assumption apparently stronger than $\W 1\ne \FPT$, namely the {\em Gap Exponential-Time Hypothesis} (Gap-\ETH), the existence of any \FPT-approximation algorithm for $k$-Clique has also been ruled out~\cite{ChalermsookCKLM20}.
    That is, $k$-Clique is {\em totally \FPT-inapproximable} -- there exists no $f(k)$-approximation algorithm for the parameterized $k$-Clique problem that runs in $t(k)\poly(n)$-time, for any computable function $f$ and $t$, depending only on $k$ with $k/f(k)$ unbounded and non-decreasing -- unless Gap-\ETH\ is false.
    It has been a major open problem whether the standard parameterized intractability assumption (i.e., $\mathrm{W}[1]\neq\mathrm{FPT}$) also implies that the $k$-Clique problem is totally FPT-inapproximable.

    In a breakthrough work~\cite{Lin21}, Bingkai Lin show that it is $\W 1$-hard to approximate the parameterized $k$-Clique problem within any constant factor.
    Later on, Karthik~C.S.\ and Khot~\cite{KarthikK22} improve the inapproximation ratio to $k^{o(1)}$ based on the same framework of Lin.
    The same lower bound is also recently obtained under \ETH\ by Lin, Ren, Sun, and Wang~\cite{LinRSW22-Clique}. 
    All these results start from the $k$-vector-sum problem ($k$-Vector Sum for short), which might be viewed as an algebraization of $k$-Clique. Then $k$-Vector Sum is reduced to some appropriate CSP problem where 
    gap between Yes- and No-instances emerges by some tricky algebraic techniques. Finally, by applying the renowned FGLSS reduction~\cite{FeigeGLSS96}, we obtain $k$-Clique instances with the desired gap.
    
    In this paper, we show the same $k^{o(1)}$-appproximation lower bound for $k$-Clique by a relatively simple proof. Although our gap-creating reduction is inspired by  the framework of Lin~\cite{Lin21},
    we deviate from Lin by re-interpreting his construction as the composition of {\em Sidon Sets} and {\em Network Coding} applied directly to the $k$-Clique problem.

    More precisely, the Sidon set is a set $S$ of vectors such that any two distinct vectors in $S$ sum to distinct values.
    This allows us to ``label'' vertices in a graph with vectors so that the sum of any two adjacent vertices yields a unique ``edge-label.''
    The second ingredient is the network coding technique that we use to compress the information of any $k$ vertices of the input graph as linear combinations of vectors from the Sidon set.
    Each linear combination becomes a ``node'' in the resulting graph.
    Now, if two linear combinations have coefficients that differ by one or two positions, then we can subtract them to  determine the vertex or edge they encode. 
    \yijia{Removed according to reviewer C: Thus, it is sufficient to randomly pick $k$ linear combinations whose coefficients are pairwise differ by two positions and decode the $k$-clique.
    Since the corresponding coefficient vectors are linearly independent, it is also solved to a unique solution.} 
    As every node encodes $k$ vertices of the input graph, we can arbitrarily pick each node as the center of the group and determine the validity of the encoding. This enables us to show the existence of some constant gap between Yes- and No-instances. In order to achieve super constant gap, we use the linear form of Sidon sets, i.e., {\em linear Sidon sets} (see Section~\ref{subsec:sidon-sets} for more details). Besides linear Sidon sets and network coding, other parts of our proof are combinatorial.

    
     As already mentioned, our construction gives the same inapproximation ratio as that in~\cite{KarthikK22} by Karthik~C.S.\ and Khot. However, our reduction has better parameters, which might be important for other applications. In fact, \cite{KarthikK22} presents a reduction from $k$-Vector Sum to $(q^{2k^2}, q^{2k^2- 1/k})$-Gap Clique (see Section~\ref{sec:prelim} for the precise definition of Gap Clique), where $q$ is a prime greater than $2^{12k}$.     
     Composing it with the reduction from $k$-Clique to $k$-Vector Sum~\cite{Lin21}, which incurs a quadratic blowup in the parameter, we obtain a self-reduction from $(k, k-1)$-Gap Clique to $(q^{\Theta(k^4)}, q^{\Theta(k^4-1/k^2)})$-Gap Clique, where $q= 2^{\Omega(k^2)}$ is a prime. On the other hand, our self-reduction from $k$-Clique creates instances of $(q^{k},q^{k-1})$-Gap Clique, where $q$ is an arbitrary prime power. As a by-product, this implies another result in~\cite{LinRSW22-Clique} that there is no $t(k)n^{o(\log k)}$-time approximation algorithm for $k$-Clique with constant approximation ratio unless \ETH\ fails.
     Furthermore, we believe that our reduction has the potential to be applied recursively to obtain an arbitrarily large gap.
     Very recently, Lin~et~al. \cite{LinRSW23} proposed a new technique for obtaining a constant-inapproximability with a polynomial parameter blow-up under ETH. However, the constrained satisfaction problem (CSP) that their technique produces requires arity at least three, which means that it is not applicable to prove the inapproximability result under the $\mathrm{W}[1]$-hardness, due to a specific technicality. Moreover, the gap-amplification is limited at $O(\log k/\log\log k)$.
     %

    The comparison of parameter transformation in all the known reductions is shown in~\Cref{tab:comparison}. 

     \begin{table}
         \centering
         \begin{tabular}{c|c|c|c|c}
              Work & Hardness Source &  Parameter Change & Inapprox Ratio & Running Time \\
              \hline
              \cite{Lin21} & $k$-Vector Sum
                           & $K=2^{\Omega(k^3)}$
                           & $O(1)$
                           & $n^{\Omega(\sqrt[6]{\log K})}$\\
                             & {\tiny (from $k'$-Clique, $k=\binom{k'}{2}$)} 
                             & {\tiny for $c > 1$}
                             & {\tiny ($2^{c k^3}$ vs $2^{c k^3-1}$)}
                             & \\
                           
              \cite{KarthikK22}
                           & $k$-Vector Sum
                           & $K=q^{2k^2}$
                           & $K^{o(1)}$
                           & $n^{\Omega\left(\frac{\log K}{\log q}\right)}$\\
                             & {\tiny (from $k'$-Clique, $k=\binom{k'}{2}$)} 
                             & {\tiny for prime $q > 2^{12k}$}
                             & {\tiny ($q^{2 k^2}$ vs $q^{2 k^2-1/k}$)}
                             & \\
              \cite{LinRSW22-Clique}
                           & $k$-Vector Sum
                           & $K=q^{2k}$
                           & $K^{o(1)}$
                           & $n^{\Omega\left(\frac{\log K}{t}\right)}$\\
                             & {\tiny (from $3$-SAT)}
                             & {\tiny for $c > 1$, $t = o(\log k)$}
                             & {\tiny ($2^{c k t}$ vs $2^{ (c k - 1) t}$) }
                             & {\tiny ($t = \Theta(\log q)$ for $q=2^{ct}$)}\\
              \cite{LinRSW23}
                           & $k$-var Vector $2$-CSP
                           & $K=k^{\frac{\log k}{\log\log k}}$
                           & $K^{\Omega(1 - \frac{1}{c \log K}}$
                           & $n^{\Omega\left(\frac{\log K}{\log\log K}\right)}$\\
                             & {\tiny (from $3$-SAT)}
                             & {\tiny for $c > 1$, $t = o(\log k/\log\log k)$}
                             & {\tiny ($k^{t}$ vs $k^{c t}$) }
                             & {\tiny ($t = \Theta(\log q)$ for $q=2^{ct}$)}\\
              This paper
                           & $k$-Clique
                           & $K=q^{k}$
                           & $q$
                           & $n^{\Omega\left(\frac{\log K}{\log q}\right)}$\\
                             & {\tiny (self-reduction)}
                             & {\tiny for any prime $q=\geq 2$}
                             &  {\tiny ($q^{k}$ vs $q^{k - 1}$)}                             
                             & {\tiny for any prime $q\geq 2$}\\
         \end{tabular}
         \caption{Comparison of Parameter Transformation and Running Time Lower Bound under ETH.}
         \label{tab:comparison}
     \end{table}\yijia{Old: in [Lin21] parameter change: $K= 2^{\Omega(k)}$, inapprox ratio: $2^{ck}$ vs $2^{ck-1}$, running time $n^{\Omega(\log K)}$. In [LRSW22], parameter change: $K=q^{2k}$.
     Changed according to reviewer B}


    \section{Preliminaries}
    \label{sec:prelim}
    
    We use standard terminology from graph theory. Let $G=(V,E)$ be a graph, where $V= V(G)$ is the vertex set and $E= E(G)$ is the edge set. An edge in $G$ between two distinct vertices $u,v\in V$ is either denoted by $uv\in E$ or $\{u,v\}\in E$ whichever is convenient.  
    A \emph{$k$-clique} is a complete subgraph of $G$ on $k$ vertices.
    We may refer to a clique in $G$ using a subset of vertices $S\subseteq V(G)$ that induces a complete subgraph. 

    In the $k$-Clique problem, we are given a $G$ and an integer $k\ge 1$, and the goal is to determine whether $G$ has a clique of size at least $k$.
    The \emph{Maximum Clique} problem is a maximization variant of $k$-Clique, where we are asked to find a maximum-size clique in $G$.
    However, we will mostly abuse the name $k$-Clique also for its maximization variant, and here $k$ denotes the optimal solution, i.e., the size of the maximum clique in $G$.\footnote{As far as \FPT-approximation algorithms are concerned, two versions are indeed equivalent~\cite{CGG07}.}
    When we refer to the $k$-Clique problem in the context of parameterized complexity, we mean the $k$-Clique problem parameterized by $k$, i.e., the parameter $k$ is a small integer independent of the size of $G$.

    The $(k,k')$-Gap Clique problem is a promise version of $k$-Clique that asks to decide whether the graph $G$ has a clique of size $k$ or every clique in $G$ has size at most $k'$.
    Again, $k$ is the standard parameter used when discussing parameterized algorithms and parameterized complexity.

    \medskip
    For a prime power $q\in \mathbb N$ we use $\mathbb F_q$ to denote the finite field with $q$ elements. Let $d\ge 1$. As usual, $\mathbb F_q^d$ is the vector space over $\mathbb F_q$, which consists of all vectors of the form
    \[
    r= (r_1, \ldots, r_d)
    \]
    where $r_i\in \mathbb F_q$ for every $i\in [k]$. We also use $r[i]$ to denote the $i$-th coordinate of $r$, i.e., $r[i]= r_i$. Recall that the standard unit vector $e_i\in \mathbb F_q^d$ is defined by
    \begin{equation}\label{eq:ei}
    e_i[i']= \begin{cases}
        1 & \text{if $i'= i$} \\
        0 & \text{otherwise}.
    \end{cases}
    \end{equation}
    Let $r, r'\in \mathbb F_q^d$. We define the set of coordinates (or positions) on which $r$ and $r'$ differ by
    \[
    \diff(r, r')\coloneqq \big\{i\in [d] \bigmid  r[i]\ne r'[i]\big\}.
    \]
    Then the \emph{Hamming distance} between $r$ and $r'$ is 
    \[
    \dist(r,r')
     \coloneqq \big|\diff(r,r')\big|.
    \]

    \subsection{Parameterized Complexity}
    \label{subsec:fpt-complexity}

    We follow the definitions and notations from~\cite{flumG06-FPTBook}.
    In the context of computational complexity, a decision problem is defined as a set of strings over a finite alphabet $\Sigma$, sometimes called a {\em language}, say $\Pi\subseteq\Sigma^*$.
    A {\em parameterization} of a problem is a polynomial-time computable function $\kappa:\Sigma^*\rightarrow \mathbb{N}$.
    A {\em parameterized decision} problem is then defined as a pair $(\Pi,\kappa)$, where $\Pi\subseteq \Sigma^*$ is an arbitrary decision problem, and $\kappa$ is its parameterization.

    The parameterization $\kappa$ gives a characterization of an instance of the designated problem that depends on the instance's property but not on the input size.
    That is, the function $\kappa$ maps an instance of a problem to a small positive integer, and we may think of a parameterized decision problem as a problem where each input instance $x\in\Sigma^*$ is associated with a number $\kappa(x)$ which is typically much smaller than the length $|x|$ of $x$.

    We say that a parameterized problem $(\Pi,\kappa)$ is {\em fixed-parameter tractable} (\FPT) if there is an algorithm that, given a string $x\in \Sigma^*$ decides whether $x\in \Pi$ in time $t(\kappa(x))\poly(|x|)$, where 
    $\poly(n)= \bigcup_{c>0}n^c$ and $t(k)$ is a computable function that depends only on the parameter $k= \kappa(x)$.
    The running time of the form $t(\kappa(x))\poly(|x|)$ is called $\FPT$-time, and the algorithm that decides a parameterized problem in $\FPT$-time is called an \emph{$\FPT$ algorithm}.
    The complexity class $\FPT$ is the class of all parameterized problems that admit $\FPT$ algorithms.
    Similar to the polynomial-hierarchy in the theory of $\NP$-completeness, there exists the classes of W-hierarchy such that
    \[
        \FPT\subseteq \W 1 \subseteq \W 2 \subseteq \cdots \subseteq \W i\subseteq \cdots.
    \]
    The classes $\W 1$ and $\W 2$ contain many natural complete problems, in particular  the $k$-Clique problem and the $k$-Dominating Set problem, respectively.
    Thus, $k$-Clique admits no $\FPT$ algorithm unless $\W 1= \FPT$, and similarly, $k$-Dominating-Set admits no $\FPT$ algorithm unless $\W 2= \FPT$ (which thus implies $\W 2= \W 1 = \FPT$).

    $\FPT$ algorithms have been one of the tools for coping with \NP-hard problems as when the parameter $\kappa(x)$ is much smaller than $|x|$, an $\FPT$ algorithm is simply a polynomial-time algorithm.
    However, as mentioned, assuming $\W 1\ne \FPT$, many important optimization problems like the maximum clique problem and the minimum dominating set problem admit no $\FPT$ algorithm. This leads to the seek of an {\em $f(k)$-\FPT-approximation algorithm} for those optimization problems, where $f$ is a computable function. That is, an \FPT-time algorithm that, given an instance $x$ with the additional parameter $k$ such that cost of an optimal solution, denoted by $\opt(x)$, satisfies
    \[
    \begin{cases}
    \opt(x)\ge k & \text{for maximization problems} \\
    \opt(x)\le k & \text{for minimization problems},
    \end{cases}
    \]
    produces a feasible solution $y$ with
    \[ 
    \begin{cases}
    \cost(y)\ge \frac{k}{f(k)}
     & \text{for maximization problems} \\
    \cost(y)\le k\cdot f(k)
     & \text{for minimization problems}.
    \end{cases}
    \]
    The function $f$ is known as the \emph{approximation ratio} of the algorithm. 
    It should be clear that, for maximization problems, we are only interested in that ratio $ f$ with unbounded and non-decreasing $k/ f(k)$.
    

    The development of \FPT-approximation algorithms has been at a fast pace in both upper bound and lower bound (i.e., the \FPT-inapproximability results).
    Please see, e.g.,~\cite{FeldmannKLM20-Survey} for references therein.
    At present, two core problems in the area of parameterized complexity -- the $k$-Clique problem and the $k$-Dominating Set problem -- are known to be \emph{totally \FPT-inapproximable}, i.e., they admit no \FPT-approximation algorithm parameterized by the size of optimal solution $k$ for any ratio, unless Gap-\ETH\ is false~\cite{ChalermsookCKLM20}, and the total \FPT-inapproximability of the $k$-Dominating Set problem under $\W 1\ne \FPT$ was later proved by Karthik~C.S, Laekhanukit and Manurangsri in~\cite{KarthikLM19}.
    There has been steady progress on the  hardness of approximation for $k$-Clique and $k$-Dominating Set~\cite{LinRSW22-SetCover}.
    Nevertheless, the question of whether $k$-Clique is totally \FPT-inapproximable under $\W 1\ne \FPT$ remains open.

    \yijia{The rest part of this section seems not useful for our current results.}
    To be formal, we say that a maximization (respectively, minimization) problem $\Pi$ parameterized by a standard parameter $k$ (e.g., the size of the optimal solution) is {\em totally FPT-inapproximable} if, for any computable non-decreasing functions $t(k)$ and $f(k)$, depending only on $k$, given an input of size $n$, there exists no algorithm running in time $t(k)\poly(n)$ that outputs a feasible solution with the cost at least $k/f(k)$ (respectively, at most $f(k)\cdot k$ for minimization problems).

    \begin{openproblem}[${\mathrm{W}[1]}$-hardness of $k$-Clique]
        \label{open:clique}
        Does the total FPT-inapproximability of the $k$-Clique problem hold under $\mathrm{W}[1] \neq \mathrm{FPT}$?
    \end{openproblem}

    A similar problem is also open for the $k$-Dominating Set problem.

    \begin{openproblem}[${\mathrm{W}[2]}$-hardness of $k$-Dominating Set]
        \label{open:domset}
        Does the total FPT-inapproximability of the $k$-Dominating Set problem hold under $\mathrm{W}[2] \neq \mathrm{FPT}$?
    \end{openproblem}

    Lastly, another conjecture has been established as a generic tool for deriving FPT-inapproximability result, namely the {\em Parameterized Intractability Hypothesis} (PIH).
    The conjecture involves asking whether a {\em $2$-Constraint Satisfaction} problem ($2$-CSP) on $k$ variables,
    where $k$ is a constant, while each variable, say $X_i$, takes a value from $[n]$ admits no constant FPT-approximation algorithms.
    More formally, $2$-CSP is a constraint satisfaction problem in which each constraint involves exactly two variables, and it is hypothesized that there is no constant-factor approximation for such a $2$-CSP parameterized by $k$ in FPT-time:

    \begin{hypothesis}[Parameterized Intractability Hypothesis]
        For some constant $\epsilon>0$, there is no $(1+\epsilon)$-factor FPT-approximation algorithm for $2$-CSP on $k$ variables $2$-CSP on $k$ variables parameterized by $k$, where each variable takes value from $[n]$.
    \end{hypothesis}

    The result of Chalermsook~et~al.\cite{ChalermsookCKLM20}
    implies that PIH holds under Gap-ETH, which has been improved to merely assuming \ETH\ in a very recent breakthrough~\cite{ETHPIH}.
    Nevertheless, it has been a major open problem whether PIH holds under the $\mathrm{W}[1]$-hardness:

    \begin{openproblem}[Parameterized Intractability Hypothesis]
        \label{open:PIH}
        Does the Parameterized Intractability Hypothesis hold under $\mathrm{W}[1]\neq \mathrm{FPT}$?
    \end{openproblem}
    
    \subsection{Network Coding}
    \label{subsec:network-coding}
    
    Network coding is an information compression technique used in multi-input multicast networks. This approach was first introduced in~\cite{AhlswedeCLY00-net-info-flow}, formalized in \cite{HoLMKCE03-netcoding} for delay-free acyclic networks and then generalized to the general case in \cite{Karger:NCoding1}. Please see \cite{YuengLC-Netcoding-Book} for references therein. 
    
    In the model of {\em random linear network coding}, the data are transmitted from $k$ different source nodes to $k$ different destinations as $k$ vectors over the finite field $\mathbb{F}_q^d$, say $v_1,\ldots,v_k \in {\mathbb{F}}^d_q$, where $q$ is a sufficiently large prime power. Whenever data packets meet at an intermediate node, the data are compressed and transmitted as a linear combination of the $k$ vectors with random coefficients, i.e., 
    \[
    r_1v_1 + r_2v_2 + \ldots + r_kv_k,
    \]
    where $r_i\in \mathbb F_q$ for every $i\in[k]$
    
    Given that a sink node receives enough packets, it can decode the information correctly by solving the linear system. More precisely, in the linear network coding with random coefficients, it was proved that $O(k)$ packets are enough to guarantee the existence of $k$ linearly independent vectors. Thus, there is a unique solution to the linear system, allowing the sink nodes to retrieve and correctly decode the information. 

    The advantage of network coding is in the efficiency of the throughput, which is close to the optimum, and the amount of memory required to store the packets. 
    
    The connection between the network coding and the hardness of approximating $k$-Clique is not known until the breakthrough result of Lin~\cite{Lin21}, who implicitly applied the network coding approach to compress the information of $k$ vertices as one single vector. Once we have $k$ vectors with linear independent coefficients, then the information of the original $k$ vertices can be decoded uniquely, thus allowing one to check whether the encoded $k$ vertices form a $k$-clique or not. 
    As every vector encodes information of $k$ vertices, every subset of $k$ linearly independent vectors gives a unique solution to the linear system, making the verifier reject a graph that has no $k$-clique by reading only a tiny portion of the vectors. This, consequently, gives an approximation hardness of the $k$-Clique problem through a very clever reduction from this specific probabilistic proof system to an instance of $k$-Clique.

    Technically, the network coding is the same as Hadamard code, but it is more in line with our intuition of the self-reduction creating the gap for the $k$-Clique problem. See Section~\ref{sec:overview} for more detailed discussions.
    
    \subsection{Sidon Sets and Generalization}
    \label{subsec:sidon-sets}

    %
    %
%

    A \emph{Sidon set} is a subset $S$ of an abelian group
    such that the sum of any two distinct elements in $S$ are different, i.e., for any $x,y,x',y' \in S$ with $x\ne y$ and $x'\ne y'$, we have $x + y = x' + y'$ if and only if $\{x, y\}=\{x', y'\}$.
    Given a positive integer $n\ge 1$, a Sidon set $S\subseteq \mathbb Z/n\mathbb Z$ of size at least $n^{1/3}$ can be constructed in polynomial time using a greedy algorithm (the algorithm is attributed to Erd\"{o}s): Start with a set $A= \emptyset$.
    Iteratively add to $A$ a number $x\in \{0, \ldots, n-1\}\setminus A$ such that $x$ cannot be written as $a + b - c$, for any $a,b,c\in A$ until no such number exists.
    Erd\"{o}s and Tur\'{a}n~\cite{ErdosT1941} showed that a Sidon set of size $\sqrt{n}$ can also be constructed efficiently using quadratic residues.
    In particular, they showed that the set $S= \{2pk + (k^2): 1 \leq k \leq p\}$ is a Sidon set, where $p$ is a fixed prime, and all the operations are done under $\mathbb{F}_{p}$.
    Observe that if the vertices of a graph are from a Sidon set, then the endpoints of every edge sum to a unique element in the underlying abelian group. Thus, representing a graph vertices with a Sidon set makes it open to standard arithmetic operation, albeit under some finite fields.

    There are a few generalizations of Sidon sets~\cite{OBryant2004}.
    For our purposes, we require the linear form of Sidon sets, which have been studied in the works of Ruzsa~\cite{Ruzsa1993,Ruzsa1995}, which we call {\em linear Sidon sets}.

    \begin{definition}
        \label{def:eSidon}
        Let $\mathbb F$ be a finite field and $d\ge 1$.
        A subset $S\subseteq \mathbb
        F^d$ is an \emph{Linear Sidon set} if for all $a,b\in \mathbb F^* (=
        \mathbb F\setminus \{0\})$ and $x,y,x',y'\in S$ with $x\ne y$ and $x'\ne y'$
        we have
        \begin{eqnarray*}
            ax+by= ax'+by' & \Longrightarrow &
            \{x,y\}= \{x',y'\}.
        \end{eqnarray*}
    \end{definition}

    The above definition is a natural generalization of the {\em sum-free} set~\cite{CameronErdosConjecture,Green2004-CameronErdos} and the set excluding no three-terms arithmetic progression.
    Our linear Sidon set, on the other hands, is a special case of the $(M,b)$-free set~\cite{Green2005-reg-lemma,Shapira09} defined as a set $S$ in which there exists no non-trivial solution to a linear system $Mx = b$ takes value from $S^{\ell}$.

    \section{Overview of Our Proof}\label{sec:overview}

    In this section, we outline the key ideas and steps in our construction. The input to our reduction is a graph $G$ and an integer $k \ge 1$. For technical reasons, we assume without loss of generality that $G$ is an instance of the multi-colored $k$-Clique problem, where
    \[
    V(G)= \biguplus_{i\in [k]} V_i,
    \]
    and each vertex set $V_i$ is an independent set in $G$. Thus, any $k$-clique in $G$ must include exactly one vertex from each $V_i$. Our objective is to construct a graph $H$ from $G$ and $k$ such that there is a significant gap between the size of the maximum clique in $H$ for the {\bf Yes-Instance} (where $G$ has a $k$-clique) and the {\bf No-Instance} (where $G$ has no $k$-clique).

    \paragraph{Informal Discussions.} Let us begin with an informal explanation of the intuition behind our reduction.
    
    Generally, the gap between the Yes and No cases can be created using the $\ell$-wise graph product for some $\ell \ge 1$, generating the graph $G^\ell$ where each node in $G^\ell$ represents a clique of size $\ell$ in $G$. Thus, any $k$-clique in $G$ translates into a clique of size $k^\ell$ in $G^\ell$, creating a gap of $k^\ell$ versus $(k-1)^\ell$. However, achieving a constant gap requires $\ell$ to be at least $\Omega(k)$ (e.g., $\ell = k/2$), which results in a graph $G^\ell$ with size $n^{O(k)}$. This makes a standard graph product or self-reduction insufficient to achieve constant inapproximability under FPT-reduction.

    To overcome this, we employ a technique from network coding theory to compress the representation of $G^\ell$. Specifically, we encode each subset of $k$ vertices in $G$ as a linear combination of $k$ vertices represented by vectors over a finite field $\mathbb{F}^d_q$. Each ``node'' in the new graph corresponds to a linear combination $\pi = {\bf r}^T{\bf v}$, where ${\bf v}$ is a column vector of $k$ vertices.

    The new graph consists of $q^k$ groups of nodes, or {\em color classes}, each corresponding to the coefficients of the linear combination. Each color class forms an independent set of nodes representing all possible linear combinations. Thus, any algorithm can select at most one node from each color class to form a clique. The edges between nodes from different color classes encode consistency between two linear combinations: an edge exists between two nodes if both encode the same ``valid'' information.
    
    In network coding, the decoding is straightforward; the receiver node waits until $k$ linearly independent combinations arrive. Senders cannot deceive by sending inconsistent linear combinations since the linear system has a unique solution. However, in the construction of the $k$-Clique instance, each edge represents the decoding of information from only two linear combinations. In proof systems, the $k$-Clique instance can be viewed as a {\em clique-test} involving two-query tests performed $\binom{k}{2}$ times. Each two-query test decodes only two nodes and is effective only when the two linear combinations differ by at most two coordinates. If they differ by one coordinate, the subtraction yields a numerical encoding of a vertex in the original graph. If they differ by two coordinates, it produces the sum of two numbers, which encodes two different nodes and purportedly an edge. This creates a uniqueness issue, as a single equation with two variables may not have a unique solution. We resolve this by using a Sidon Set from additive combinatorics, which ensures that the addition or subtraction of any two elements in the set is unique, thereby addressing the non-uniqueness problem.

    This approach departs from previous works \cite{Lin21} and \cite{KarthikK22}, which used a similar encoding but treated it as {\em Hadamard} codes. Karthik and Khot \cite{KarthikK22} resolved the uniqueness issue by leveraging properties of Hadamard codes, arguing that most two-query tests have unique solutions. In contrast, we rely on the Sidon Set's properties to resolve these issues.

    Our proof further diverges by viewing the encoding through the lens of network coding, using linear independence to argue consistency. Specifically, it is not guaranteed that two different nodes encode the same information—an issue addressed in locally testable codes through {\em local codeword testing}. However, we bypass local testing by arguing via linear independence, a technique applicable exclusively to the clique problem. For any node $v$, we treat it as the center of a ball, adding $k$ nodes whose coefficients differ by one position (within Hamming distance one). The $k$ nodes are linearly independent, ensuring a unique solution, preventing adversaries from feeding inconsistent information. This allows us to {\bf completely bypass local testing}, an advantage unique to the clique problem.

    More concretely, our argument proceeds as follows. If the input graph contains a $k$-clique, then there exists a clique of size $q^k$. Otherwise, let $Q$ be a clique provided by any algorithm. We know that each node in $Q$ is from a different color class. Selecting any node $v$ as a center, we consider the nodes around it within Hamming distance one (with respect to coordinates). As discussed, these nodes form a linearly independent set, yielding a unique solution. Moreover, any two non-center nodes that are in different directions from the center are at Hamming distance two, triggering an {\em edge test}. An edge exists between these two nodes if their subtraction decodes to an edge in the original graph.

    Combining the consistency property with the uniqueness of decoding, we conclude that each ball has non-center nodes in all $k$ directions away from the center if and only if their decoding forms a clique in the original graph. Therefore, we establish the following: (1) there are $q^k$ centers, leading to $q^k$ balls, and (2) each ball misses at least one direction, resulting in the loss of $q$ balls in that direction. Using a {\em two-term} Sidon Set, we show that each of the missing $q$ balls can belong to at most $k$ balls, implying that if the original graph has no $k$-clique, the resulting graph can have a clique of size at most $q^k \cdot k / q = q^{k-1} / k$. Finally, we will show that double counting does not occur if we employ a higher-order Sidon Set, say a {\em four-term} Sidon Set.

    \paragraph{More formal description.}

    Now let us give a more formal description of our reduction at some high level.
    To begin with, we identify the vertex set $V(G)$ with a linear Sidon set $S$ over a finite field $\mathbb F_q^d$ for an appropriate integer $d\ge 1$. We will show in Section~\ref{sec:linear-Sidon} that such $S= V(G)$ can be constructed with the size of $\mathbb F_q^d$ polynomially bounded by $|V(G)|$. As the second step, for a fixed vector $r= (r_1, \ldots, r_k)\in \mathbb F_q^k$ every $k$-element subset $\{v_1, \ldots, v_k\}\subseteq V(G)$ with $v_i\in V_i$ for every $i\in [k]$ is associated with 
    \begin{equation}\label{eq:kelementencoding}
    \pi = r_1v_1+ r_2v_2+\cdots + r_kv_k, 
    \end{equation}   
    exactly the same way as we transmit $k$ vectors $v_1, \ldots, v_k$ by the network coding using $r_1, \ldots, r_k$ as coefficients. Thereby, for every $r\in \mathbb F_q^k$ we have a copy of $\mathbb F_q^d$ as the \emph{column $C_r$ of vertices in $H$ indexed by $r$}. 
    

    Given that the size of $\mathbb F_q^d$ is polynomial in $|V(G)|$, there are only a polynomially bounded number of  linear combinations of $k$ vectors. However, the number of $k$-element subsets of $V(G)$ is $\Omega(n^k)$. This is simply because each vector $\pi$ obtained in~\eqref{eq:kelementencoding} encodes  many different $k$-element subsets. However, if we have two $\pi$ and $\pi'$ from two columns $C_r$ and $C_{r'}$ with the Hamming distance between $r= (r_1, \ldots, r_k)$ and $r'= (r'_1, \ldots, r'_k)$ being one, i.e., $r_i\ne r'_i$ for exactly one $i\in [k]$, then $\pi- \pi'$ will give us a unique vertex $v_i\in V_i$ with 
    \begin{equation}\label{eq:vi}
    v_i= (r_i-r'_i)^{-1}(\pi- \pi').
    \end{equation} 
    Of course, it could happen that $(r_i-r'_i)^{-1}(\pi- \pi')$ is not an element in the linear Sidon set $S= V(G)$. Hence, we only add an edge between $\pi$ and $\pi'$ in case~\eqref{eq:vi} is really a vertex in $G$. In other words, the edges between columns $r$ and $r'$ test whether two encoding $\pi\in C_r$ and $\pi'\in C_{r'}$ are consistent on their $i$-th positions where $r_i\ne r'_i$. We might say that $\pi$ and $\pi'$ \emph{pass the vertex test by the vertex $v_i\in V(G)$}.
    
    Now assume that, with $r\in \mathbb F^k_q$ we have for all $i\in [k]$ an $r^i\in \mathbb F_q^k$ which differs from $r$ only on the $i$-th position. Furthermore, with $\pi\in C_r$,  assume that there are $\pi^1\in C_{r^1}, \ldots, \pi^k\in C_{r^k}$ all adjacent to $\pi$ in $H$. By the above construction, we have $v_1\in V_1, \ldots, v_k\in V_k$ with each pair $\pi$ and $\pi^i$ passing the vertex test by the vertex $v_i$. Then if $v_i$ and $v_{i'}$ are adjacent in $G$ for all $1\le i< i'\le k$, the original graph $G$ has a $k$-clique, i.e., $\{v_1, \ldots, v_k\}$. On the other hand, it is easy to see that $r_i$ and $r_{i'}$ differ exactly on two positions, i.e., $i$ and $i'$. As $v_i$ and $v_{i'}$ are both from the linear Sidon set $S= V(G)$, it implies that they are uniquely 
    determined by $\pi^i- \pi^{i'}$. Thus we add an edge between $\pi^i$ and $\pi^{i'}$ exactly when there is an edge between the corresponding $v_i$ and $v_{i'}$ in $G$. As a consequence, 
    \begin{equation}\label{eq:kclique}
    \text{if $\pi, \pi^1, \ldots, \pi^k$ induce a clique in $H$, then $G$ has a $k$-clique}.
    \end{equation}
    
    More generally, we have two $r= (r_1, \ldots, r_k)$ and $r'= (r'_1, \ldots, r'_k)$ in $\mathbb F_q^k$ with Hamming distance two, i.e., $r_i= r'_i$ for all $i\in [k]$ but two $1\le i_1< i_2\le k$. Moreover, $\pi\in C_{r}$ and $\pi'\in C_{r'}$ with 
    \begin{equation}
    \pi- \pi'= (r_{i_1}-r'_{i_1})v 
     + (r_{i_2}-r'_{i_2})v'.
    \end{equation}
    Then there is an edge between $\pi$ and $\pi'$ if and only if $v$ and $v'$ are adjacent in $E(G)$. That is, $\pi$ and $\pi'$ \emph{pass the edge test with the edge $vv'\in E(G)$}.
    
    To summarize, the graph $H$ consists of columns $C_r$ of vertices indexed by $r\in \mathbb F_q^k$. Each column $C_r$ is a copy of $\mathbb F_q^d$ which is supposed to encode a $k$-element subset of $V(G)$ by the linear combination of the form~\eqref{eq:kelementencoding}. The edges between two columns $C_r$ and $C_{r'}$ depend on the Hamming distance between $r$ and $r'$. They correspond to the vertex test and the edge test if the distance between $r$ and $r'$ is one or two, respectively. If it is more than two, then we will add all the edges between $C_r$ and $C_{r'}$. Figure~\ref{fig:H} illustrates a part of the construction of $H$.
    
    \begin{figure}
    \centering
    \begin{tikzpicture}[
  scale=.3,
  vertex/.style={circle,inner sep=0pt,minimum size=1mm,fill=black},
  ]


\begin{scope}

    \filldraw[fill opacity=0.2, gray] (0,0) ellipse (2cm and 3cm);
    \path (0,-4.5) node[above] {\footnotesize $C_{r}$};
    \node[vertex] (pi) at (0,1.2) {};
    \path (pi) node[below] {\footnotesize $\pi$};

    \filldraw[fill opacity=0.2, gray] (6,6) ellipse (2cm and 3cm);
    \path (6,9) node[above] {\footnotesize $C_{r^1}$};
    \node[vertex] (pi1) at (5.5,7.5) {};
    \path (pi1) node[right] {\footnotesize $\pi^1$};

    \filldraw[fill opacity=0.2, gray] (-6,6) ellipse (2cm and 3cm);
    \path (-6,9) node[above] {\footnotesize $C_{r^2}$};
    \node[vertex] (pi2) at (-5.8,5) {};
    \path (pi2) node[left] {\footnotesize $\pi^2$};

    \filldraw[fill opacity=0.2, gray] (14,0) ellipse (2cm and 3cm);
    \path (14,-4.5) node[above] {\footnotesize $C_{r^{\infty}}$};
    \node[vertex] (piinfty) at (14,-.5) {};
    \path (piinfty) node[above] {\footnotesize $\pi^{\infty}$};

    \draw[-] (pi)--(pi1);
    \draw[-] (pi)--(pi2);
    \draw[-] (pi1)--(pi2);
    \draw[-] (pi)--(piinfty);

\end{scope}

\end{tikzpicture}
    \caption{Let $r,r^1,r^2,r^{\infty}\in \mathbb F_q^k$ with $\dist(r,r^1)=1$, $\dist(r,r^2)=1$, $\dist(r^1, r^2)= 2$, and $\dist(r,r^{\infty})\ge 3$. More precisely, say $r$ and $r^1$ differ on their \emph{first} positions, and $r$ and $r^2$ on their \emph{second} positions, hence $r^1$ and $r^2$ differ exactly on their first and second positions. As a consequence, the edge between $\pi\in C_r$ and $\pi^1$ means that $\pi- \pi^1= (r[1]-r^1[1])v_1$ for some $v_1\in V_1$, and similarly $\pi-\pi^2= (r[2]-r^2[2])v_2$ for some $v_2\in V_2$ by the edge $\pi\pi^2$. Here, we use $r[1]$ to denote the first coordinate of the vector $r\in \mathbb F_q^k$, and similarly, $r^1[1]$ is the first coordinate of $r^1$. Furthermore, the edge between $\pi^1$ and $\pi^2$ implies that $v_1v_2$ is an edge in the original graph $G$. Finally, since $\dist(\pi, \pi^{\infty})\ge 3$, there is an edge between $\pi$ and any $\pi^{\infty}\in C_{r^{\infty}}$.}
    \label{fig:H}
    \end{figure}
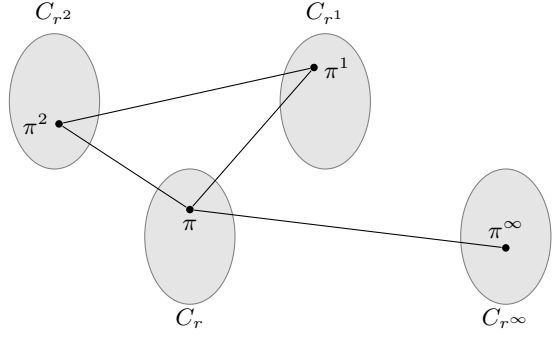
    
    If $G$ has a $k$-clique of vertices $v_1, \ldots, v_k$, then for every column $C_r$ we can pick $\pi\in C_r$ as~\eqref{eq:kelementencoding}. Our construction ensures they form a clique in $H$ of size $q^k$. Otherwise, i.e., any clique $K$ in $G$ has size at most $k-1$, we will argue that the size of a maximum clique in $G$ is on the order of $q^{k-1}$. The key observation is that, as implied by~\eqref{eq:kclique}, for every $r\in \mathbb F_q^d$, if there is a $\pi\in K\cap C_r$, then for at least one $i\in [k]$ for all $r'$ which differs from $r$ in exactly the $i$-th position, the column $C_{r'}$ contains no vertex in $K$.
    
    \section{Construction of Linear Sidon Sets}
    \label{sec:linear-Sidon}

    This section presents the construction of a linear Sidon set of a given size using a greedy algorithm similar to that of Erd\"{o}s mentioned in the Preliminaries.

    \medskip
    Fix a finite field $\mathbb F$ and $d\ge 1$. To ease the presentation, we introduce a further technical notion.
    \begin{definition}\label{def:4linearindp}
    A subset $S\subseteq \mathbb F^d$ is \emph{$4$-term linearly independent} if every $S'\subseteq S$ with $|S'|\le 4$ is linearly independent. 
    Equivalently, every $x\in S$ is not a linear combination of three (not necessarily distinct) vectors in $S\setminus \{x\}$.
    \end{definition}

   \begin{lemma}\label{lem:4linearindp}
   Let $S\subseteq \mathbb F^d$ be $4$-term linearly independent. Then $S$ is a linear Sidon set.
   \end{lemma}

   \begin{proof}
   Let $a,b \in \mathbb F^*= \mathbb F\setminus \{0\}$ and $x,y,x',y'\in S$ such that $x\ne y$, $x'\ne y'$, and $ax+ by= ax'+ by'$. Assume that
   $\{x,y\}\ne \{x',y'\}$.
   By symmetry, we can further assume without loss of generality
   \[
   x\notin \{y, x',y'\}.
   \]
   Note that
   \[
   x= a^{-1}(ax'+ by'- by)= x'+ a^{-1}by'- a^{-1}by.
   \]
   Thus $x$ is a linear combination of $x',y',y\in S\setminus \{x\}$, contradicting Definition~\ref{def:4linearindp}.      
   \end{proof}
    
    By induction, we will construct a sequence
    \[
        S_0\subsetneq S_1\subsetneq \ldots \subsetneq S_i \subsetneq \ldots \subsetneq \mathbb F^d
    \]
    of $4$-term linearly independent subsets of $\mathbb F^d$, hence linear Sidon sets by~Lemma~\ref{lem:4linearindp}. We start with $S_0\coloneqq \emptyset$ which vacuously satisfies Definition~\ref{def:4linearindp}.
    Let $i\ge 0$ and assume that the $4$-term linearly independent $S_i\subseteq \mathbb F^d$ has already
    been constructed with $|S_i|=i$. We define a set
    \begin{align*}
        \tspan(S_i) \coloneqq \big\{ax+by+cx'\bigmid
        & \text{$a,b,c\in \mathbb F$ and $x,y,x'\in S_i$}
        \big\}.
    \end{align*}
    It is easy to see
    \begin{equation*}
        |\tspan(S_i)|\le |S_i|^3\cdot |\mathbb F|^3= i^3\cdot |\mathbb F|^3.
    \end{equation*}
    Assume
    \begin{equation}
        \label{eq:sizespan}
        i^3\cdot |\mathbb F|^3 < |\mathbb F|^d = \left|\mathbb F^d\right|,
    \end{equation}
    then there is a
    \begin{equation*}
        \label{eq:w}
        w\in \mathbb F^d\setminus \tspan(S_i).
    \end{equation*}
    Choose arbitrarily such a $w$ and let
    \[
        S_{i+1} \coloneqq S_i\cup \{w\}.
    \]
    Clearly $S_{i+1}$ is again $4$-term linearly independent.
    
        \label{eq:comb}

    \medskip
    \begin{theorem}
        \label{thm:eSidon}
        Let $\mathbb F_q$ be a finite field and $n\ge 1$.
        We set
        \[
            d \coloneqq \ceil{\frac{3\log n}{\log q}+3}.
        \]
        Thus
        \[
            n\le q^{(d-3)/3}.
        \]
        Then we can construct a 4-term linearly independent set
        $S\subseteq \mathbb F_q^d$ with
        $|S|= n$ and $|\mathbb F_q^d|\le q^4\cdot n^3$ in time polynomial in $n+q$. Observe that $S$ is also a linear Sidon set by Lemma~\ref{lem:4linearindp}. 
    \end{theorem}

    \begin{proof}
        Following the greedy strategy as we described, we construct a sequence
        of 4-term linearly independent 
        \[
            S_0\subsetneq S_1\subsetneq \ldots \subsetneq S_i \subsetneq \ldots \subsetneq \mathbb F^d.
        \]
        As
        \[
            (n-1)^3\cdot |\mathbb F_q|^3<
            n^3\cdot |\mathbb F_q|^3 =
            n^3\cdot q^3 \le q^{d-3}\cdot q^3 = q^d= |\mathbb F_q|^d,
        \]
        by~\eqref{eq:sizespan} we conclude that $S_n$ can be constructed.
        Thus we can
        take $S \coloneqq S_n$ with $|S|=n$.

        \smallskip To see that the greedy algorithm runs in polynomial time, it
        suffices to observe that
        \[
            |\mathbb F_q^d|= q^d=q^{\ceil{\frac{3\log n}{\log q}+3}}\le n^3\cdot q^4. \qedhere
        \]

    \end{proof}

    \subsection{Multi-term linearly independent sets}
    For the construction in Subsection~\ref{subsec:improve}, we need to consider linear combinations of more than two vectors in a given set $S$. This leads to the following definition.
    

     \begin{definition}\label{def:teSidon}
     Let $\mathbb F$ be a finite field and $t\ge 1$. A subset $S\subseteq \mathbb F^d$ is \emph{$t$-term linearly  independent} if every $S'\subseteq S$ with $|S'|\le t$ is linearly independent.
     \end{definition}

     We leave the details of the construction of $t$-term linearly independent sets to the reader, which is a straightforward generalization of Theorem~\ref{thm:eSidon}.
     \begin{theorem}
        \label{thm:teSidon}
        Let $\mathbb F_q$ be a finite field, $t\ge 1$  and $n\ge 1$.
        We set
        \[
            d \coloneqq \ceil{\frac{(2t-1)\log n}{\log q}+ 2t- 1}.
        \]
        Thus
        \[
            n\le q^{(d-(2t-1))/(2t-1)}.
        \]
        Then we can construct a $t$-term linearly independent set $S\subseteq \mathbb F_q^d$ with
        $|S|= n$ in time polynomial in $n^t+q^t$.
    \end{theorem}

    \section{Reduction from $(k,k-1)$-Gap Clique to $(q^{k},q^{k-1})$-Gap Clique}
    \label{sec:Clique-to-Gap-Clique}


    In this section, we first present a reduction from $k$-Clique, or equivalently $(k,k-1)$-Gap Clique, to $(q^{k},k\cdot q^{k-1})$-Gap Clique. This in fact already implies that the Maximum $k$-Clique problem admits no \FPT-approximation algorithm with approximation ratio $k^{o(1)}$, unless $\FPT= \W 1$. Then we explain how to modify our construction to get a reduction from $(k,k-1)$-Gap Clique to $(q^{k}, q^{k-1})$-Gap Clique. It is less transparent than the first reduction, but will enable us to obtain a lower bound in~\cite{LinRSW22-Clique} under \ETH.

    \subsection{The reduction}
    \label{subsec:reduction}

    Let $(G,k)$ be an instance of the multi-colored $k$-Clique problem. In
    particular,
    \begin{equation}\label{eq:VG}
        V(G)= \biguplus_{i\in [k]} V_i,
    \end{equation}
    and each $V_i$ is an independent set in $G$.
    We construct a graph $H$ as
    follows.
    \begin{itemize}
        \item Let $\mathbb F_q$ be a finite field where $q$ is to be determined
        later.
        Moreover, let $n \coloneqq |V(G)|$ and $d\coloneqq \ceil{\frac{3\log n}{\log
        q}+3}$.
        So, by \Cref{thm:eSidon}, we can assume without loss of
        generality that $V(G)\subseteq \columnset$ is a linear Sidon set of
        size $n$.

        \item For every $r\in \indexset$ we define \emph{the column with index $r$}
        as
        \[
            C_r\coloneqq \Big\{\prf_{r, \pi} \Bigmid \pi\in \columnset\Big\}.
        \]
        Here $\prf_{r, \pi}$ is a unique vertex associated with $r$ and $\pi$. Then we set
        \begin{equation}\label{eq:VH}
            V(H)\coloneqq \bigcup_{r\in \indexset} C_r.
        \end{equation}


        \item We still need to define the edge set $E(H)$ for the graph $H$.
        Let
        $r, r'\in \indexset$ and $\prf_{r, \pi}\in C_r$, $\prf_{r', \pi'}\in
        C_{r'}$.
        We distinguish the following cases.
        \begin{enumerate}
            \item[(H1)] If $r=r'$, then there is no edge between $\prf_{r, \pi}$
            and $\prf_{r', \pi'}$.

            \item[(H2)] 
            If $\dist(r,r')= 1$, say $r'= r+ ae_i$\footnote{Recall $e_i$ is unit vector defined as~\eqref{eq:ei}.} for some $a\in
            \mathbb F_q^*$ and $i\in [k]$.
            Then
            \begin{align}
                \label{eq:vertextest}
                \prf_{r,\pi}\prf_{r',\pi'}\in E(H)
                & \iff 
                \text{$\pi'- \pi= av$ for some $v\in V_i\subseteq V(G)$}.
            \end{align}
            That is, $\pi$ and $\pi'$ pass vertex test with the vertex $v$ as described in Section~\ref{sec:overview}.

            \item[(H3)]
            $a,b\in \mathbb F_q^*$ and distinct $i,j\in [k]$.
            Then
            \begin{align}
                \notag
                \prf_{r,\pi}\prf_{r',\pi'}\in E(H)
                \iff & 
                \text{$\pi'- \pi= au+ bv$} \\\label{eq:edgetest}
                & \quad \text{for some $u\in V_i$ and $v\in V_j$ with $uv\in E(G)$}.
            \end{align}
            Hence, $\pi$ and $\pi'$ pass the edge test with the edge $uv$.
            Note by~Definition~\ref{def:eSidon} the vertices $u$ and $v$, if
            exist, are unique.


            \item[(H4)] If $\dist(r,r')\ge 3$, then
            $\prf_{r,\pi}\prf_{r',\pi'}\in E(H)$.
        \end{enumerate}
    \end{itemize}

    \medskip
    \begin{lemma}\label{lem:firstHtime}
        The graph $H$ can be constructed in time polynomial in $|V(G)|+q^k$.
    \end{lemma}

    \begin{proof}
    First, we identify $V(G)$ with a linear Sidon set $S\subseteq \mathbb F_q^d$. By~Theorem~\ref{thm:eSidon}, this can be done in time polynomial in $|V(G)|+ q$. Then we construct the vertex set of $H$ as~\eqref{eq:VH}, which takes time linear in the size of $V(H)$. Note
    \[
    |V(H)|= 
    |\mathbb F_q^k|\cdot |\mathbb F_q^d|= q^k\cdot q^d\le q^{k+4}\cdot |V(G)|^3.  
    \]
    Finally to construct the edge set $E(G)$ we go through each pair of vertices in $V(H)$ and check the conditions in (H1) -- (H4), which requires firstly some simple arithmetic in $\mathbb F_q^k$ and $\mathbb F_q^d$, and then checking the edge set $V(G)$. Thus, it can be done again in time polynomial in $q^k+ |V(G)|$.
    \end{proof}

    \medskip

    \subsection{The completeness}
    \label{subsec:completeness}

    \begin{lemma}
        \label{lem:completeness}
        If $G$ has a $k$-clique, then $H$ has a clique of size $q^k$.
    \end{lemma}

    \begin{proof}
        Let $\{v_1, \ldots, v_k\}\subseteq V(G)\subseteq \columnset$ be a
        $k$-clique in $G$.
        Then for every $r\in \indexset$ we define
        \[
            \pi_r\coloneqq \left(\sum_{i\in [k]} r[i]v_i[1], \ldots, \sum_{i\in [k]} r[i]v_i[d]\right)
            \in \columnset.
        \]
        We claim that
        \[
            K\coloneqq \big\{\prf_{r, \pi_r} \bigmid r\in \indexset\big\}
        \]
        is clique in $H$.
        Assume $r,r'\in \indexset$ with $r\ne r'$.
        We need to show
        $\prf_{r, \pi_r}$ and $\prf_{r', \pi_{r'}}$ are adjacent in $H$.
        \begin{itemize}
            \item {\bf Case 1:} $\dist(r,r')=1$, say $r'= r+ae_i$ for some $a\in \mathbb F_q^*$ and
            $i\in [k]$.
            Then
            \begin{align*}
                \pi_{r'}- \pi_{r}
                & = \left(\sum_{\ell\in [k]} r'[\ell]v_{\ell}[1], \ldots, \sum_{\ell\in [k]} r'[\ell]v_{\ell}[d]\right)
                - \left(\sum_{\ell\in [k]} r[\ell]v_{\ell}[1], \ldots, \sum_{\ell\in [k]} r[\ell]v_{\ell}[d]\right) \\
                & = \left(\sum_{\ell\in [k]} (r'[\ell]-r[\ell])v_{\ell}[1], \ldots, \sum_{\ell\in [k]} (r'[\ell]-r[\ell])v_{\ell}[d]\right) \\
                & = \big((r'[i]-r[i])v_i[1], \ldots, (r'[i]-r[i])v_i[d]\big)
                = \big(av_i[1], \ldots, av_i[d]\big) \\
                & = a\big(v_i[1], \ldots, v_i[d]\big) = a v_i.
            \end{align*}
            We are done by~\eqref{eq:vertextest}.

            \item {\bf Case 2:} $\dist(r,r')= 2$, in particular $r'= r+ ae_i+ be_j$ for some $a,b\in
            \mathbb F_q^*$ and $i,j\in [k]$ with $i\ne j$.
            It follows that
            \begin{align*}
                & \pi_{r'}- \pi_{r} \\
                = & \left(\sum_{\ell\in [k]} (r'[\ell]-r[\ell])v_{\ell}[1], \ldots, \sum_{\ell\in [k]} (r'[\ell]-r[\ell])v_{\ell}[d]\right) \\
                = & \big((r'[i]-r[i])v_i[1]+ (r'[j]-r[j])v_j[1],
                \ldots, (r'[i]-r[i])v_i[d]+(r'[j]-r[j])v_j[d]\big) \\
                = & a\big(v_i[1], \ldots, v_i[d]\big) + b\big(v_j[1], \ldots, v_j[d]\big)
                = a v_i+ b v_j.
            \end{align*}
            So~\eqref{eq:edgetest} implies that $\prf_{r, \pi_r} \prf_{r',
            \pi_{r'}}\in E(H)$.

            \item {\bf Case 3:} $\dist(r,r')\ge 3$. This is trivial by (H4).
        \end{itemize}
        Clearly, $|K|= q^k$.
        This finishes our proof.
    \end{proof}

    \medskip

    \subsection{The soundness}
    \label{subsec:soudness}

    \begin{lemma}
        \label{lem:soundness}
        If $G$ has no $k$-clique, then $H$ has no clique of size
        \[
            k\cdot q^{k-1}+1.
        \]
    \end{lemma}

    \begin{proof}
        Assume that $K\subseteq V(H)$ is a clique in $G$.
        For the latter purpose,
        let
        \[
            R\coloneqq \big\{r\in \indexset \bigmid K\cap C_r\ne \emptyset\big\}.
        \]
        Then by (H1), we have
        \[
            |K\cap C_r|= 1
        \]
        for every $r\in R$, and otherwise $|K\cap C_r|= 0$ for every $r\in
        \indexset\setminus R$.
        In the former case, we use
        \begin{equation}
            \label{eq:Pirpir}
            \prf_{r, \pi_r}
        \end{equation}
        to denote the \emph{unique} element in $K\cap C_r$.
        Thereby $\pi_r\in
        \columnset$.

        \medskip
        \begin{claim} Let $r\in R$.
        Then there exists an $i\in  [k]$
        such that for all $a\in \mathbb F^*_q$ we have
        \[
            r+ ae_i\notin R.
        \]
        \end{claim}
        
        \medskip
        \begin{proof}[Proof of the claim.] Towards a contradiction we assume that
        for every $i\in [k]$ there is an $a_i\in \mathbb F^*_q$ with
        \[
            r^i\coloneqq r+ a_i e_i\in R.
        \]
        By~\eqref{eq:Pirpir} there is a unique
        \[
            \prf_{r^i, \pi_{r^i}}\in K.
        \]
        For the same reason, we have a unique $\prf_{r, \pi_r}\in K$.
        As $K$ is a
        clique, there is an edge between $\prf_{r, \pi_r}$ and $\prf_{r^i,
        \pi_{r^i}}$ in the graph $H$.
        By (H2) we have a vertex $v_i\in V_i$ in the
        original graph $G$ such that
        \[
            \pi_{r^i}- \pi_{r}= a_{i} v_{i}.
        \]
        Now we show $\{v_1, \ldots, v_k\}$ is a $k$-clique in $G$, contradicting our
        assumption that $G$ has no $k$-clique.
        So we need to demonstrate that
        $v_{i}v_{j}\in E(H)$ for every distinct $i,j\in [k]$.
        Observe that
        \[
            \pi_{r^i}- \pi_{r^j}
            =(\pi_{r^i}- \pi_{r}) - (\pi_{r^j}- \pi_{r})
            = a_{i}v_{i} - a_{j}v_{j} = a_{i}v_{i} + (-a_{j})v_{j}.
        \]
        By (H3) we conclude $v_{i}v_{j}\in E(G)$ as desired. Let us emphasize that $V(G)$ is a linear Sidon set, hence the above $v_i$ and $v_j$ are uniquely determined. \yijia{new sentence according to Reviewer C.}
        \end{proof}

        \medskip
        Of course for each $r\in R$ we might have more than one $i\in [k]$ satisfying
        the above claim.
        Nevertheless, we fix an arbitrary one and denote it by
        $i_r$.
        Then we define
        \begin{align}
            \label{eq:Tr}
            T_r & \coloneqq \Big\{r+a\cdot e_{i_r}\in \indexset \Bigmid a\in \mathbb F^*_q\Big\}.
        \end{align}
        Note $|T_r|=q-1$ and $T_r\cap R= \emptyset$.

        \medskip
        \begin{claim}
        Every $r\in \indexset \setminus R$ can occur in
        at most $k$ many different $T_{r'}$ for $r'\in R$.
        More precisely,
        \[
            \big|\{r'\in R\mid r\in T_{r'}\}\big|\le k.
        \]
        \end{claim}

        \medskip
        \begin{proof}[Proof of the claim.] Assume $r\in T_{r'}$ and $r\in
        T_{r''}$ for distinct $r', r''\in R$.
        We show
        \[
            i_{r'}\ne i_{r''},
        \]
        which immediately implies the claim by $i_{r'}, i_{r''}\in [k]$.
        By assumption for some $a, b\in \mathbb F^*_q$ we have
        \begin{eqnarray*}
            r= r'+a e_{i_{r'}}
            & \text{and} &
            r= r''+b e_{i_{r''}}.
        \end{eqnarray*}
        So if $i_{r'}= i_{r''}$ we would have
        \[
            r''= r'+(a-b) e_{i_{r'}}.
        \]
        Note $a-b\ne 0$ as $r'\ne r''$.
        Then by our definition of $T_{r'}$ (i.e.,
        ~\eqref{eq:Tr} where $r\mapsto r'$) we conclude
        \[
            r''\in T_{r'}.
        \]
        On the other hand, $T_{r'}\cap R= \emptyset$, contradicting $r''\in R$.
        \end{proof}

        \medskip
        Now, let us continue the proof of \Cref{lem:soundness}. Putting all the pieces together, we have
        \begin{align*}
            & R\uplus \bigcup_{r\in R} T_r\subseteq \indexset & \text{\big(by $T_r\cap R= \emptyset$ for every $r\in R$\big)} \\
            & \Longrightarrow\  |R|+ \left|\bigcup_{r\in R} T_r\right| \le q^k \\
            & \Longrightarrow \ |R|+ \frac{|R|\cdot (q-1)}{k} \le q^k
            & \text{\big(by $|T_r|=q-1$ and Claim~2\big)} \\
            & \Longrightarrow \frac{|R|\cdot (q-1+k)}{k} \le  q^k \\
            & \Longrightarrow \frac{|R|\cdot q}{k}\le  q^k
            & \text{\big(by $k\ge 1$\big)} \\
            & \Longrightarrow |R|\le k\cdot q^{k-1}.
        \end{align*}
        This finishes the proof.
    \end{proof}

    \subsection{Improvement to $(q^{k},q^{k-1})$-Gap Clique}\label{subsec:improve}
    Again we start from a multi-colored $k$-Clique instance $(G, k)$ with $V(G)$ satisfying ~\eqref{eq:VG} and modify the construction of $H$ in~Subsection~\ref{subsec:reduction} as follows. 

        \begin{itemize}
        \item We identify the vertex $V(G)$ with a $8$-term linearly independent set in $\mathbb F_q^d$ as stated in Theorem~\ref{thm:teSidon}.

        \item Again the vertex set of $H$ is
        $V(H)= \bigcup_{r\in \indexset} C_r$,
        where $C_r= \Big\{\prf_{r, \pi} \Bigmid \pi\in \columnset\Big\}$ for every $r\in \indexset$.

        \item For the edge set $E(H)$ of $H$, let
        $r, r'\in \indexset$ and $\prf_{r, \pi}\in C_r$, $\prf_{r', \pi'}\in
        C_{r'}$. Set $t:= \dist(r, r')$. There is an edge between $\prf_{r, \pi}$ and $\prf_{r', \pi'}$ if and only if one of the following conditions is satisfied.
        \begin{itemize}
            \item {\bf Case 1.} $1\le t\le 4$. Assume $\diff(r,r')= \{i_1, \ldots, i_t\}$, then
            \[
            \pi- \pi' 
            = \sum_{j\in [t]} \big(r[i_j]- r'[i_j]\big)v_j,
            \]
            where $v_1\in V_{i_1}, \ldots, v_t\in V_{i_t}$ and $\{v_1, \ldots, v_t\}$ is a $t$-clique in $G$. Observe that, as $V(G)$ is $8$-term linearly independent and $t\le 4$, it is easy to see that the vertices $v_1, \ldots, v_t$ are unique. For the later purpose, we write
            \begin{equation}\label{eq:vs}
            \vs_{r,r'}(\pi, \pi'):=
             \{v_1, \ldots, v_t\}.
            \end{equation}
            
            \item {\bf Case 2.} $t\ge 5$.
        \end{itemize}
    \end{itemize}
    We remark that the above graph $H$ is an induced subgraph of our original $H$. Exactly as Lemma~\ref{lem:firstHtime} we can show:

    \medskip
    \begin{lemma}\label{lem:Htime}
        The graph $H$ can be constructed in time polynomial in $|V(G)|+q^k$.
    \end{lemma}

    Now we are ready to prove the gap between $q^k$ and $q^{k-1}$. 
    \begin{lemma}\label{lem:qgap}
    \begin{enumerate}
        \item {\bf Completeness.} If $G$ has a $k$-clique, then $H$ has a clique of size $q^k$.

        \item {\bf Soundness.} If $G$ has no $k$-clique, then $H$ has no clique of size $q^{k-1}+1$.
    \end{enumerate}
    \end{lemma}

    \begin{proof}
    The completeness case follows the same line as Lemma~\ref{lem:completeness}. For the soundness, some extra work is needed along the line of the proof of Lemma~\ref{lem:soundness}. 
    Assume that $G$ does not have a $k$-clique, and consider any clique $K\subseteq V(H)$ in $H$. Again, let $R \coloneqq \big\{r\in \indexset \bigmid K\cap C_r\ne \emptyset\big\}$.

    \begin{claim}
    Let $r\in R$.
    Then there exists an $i\in [k]$
    such that for all $r'\in \indexset$ with $\dist(r, r')\le 2$ and $i\in \diff(r, r')$ we have
    \[
    r'\notin R.
    \]
    \end{claim}

    \begin{proof}[Proof of the claim.] Assume that for every $i\in [k]$ there is an $r^i\in R$ with
    \begin{eqnarray}\label{eq:rri}
        \dist(r,r^i)\le 2 & \text{and} & i\in \diff(r, r^i).
    \end{eqnarray}
    Fix such an $r^i$. 
    Since $r^i\in R$, there is a unique $\prf_{r^i, \pi_{r^i}}\in K$. Similarly we have a unique $\prf_{r, \pi_r}\in K$. It follows that $\vs_{r, r^i}(\pi_r, \pi_{r^i})$ as defined in~\eqref{eq:vs}, which has size at most $2$, contains a vertex
    \[
    v_i\in V_i.
    \]
    Now let $1\le i< j\le k$. Then by~\eqref{eq:rri} we have $\dist(r^i,r^j)\le 4$. Furthermore, it is routine to verify that either $\{i,j\}\subseteq \diff(r^i, r^j)$, or $\{i,j\}\subseteq \diff(r, r^i)$, or $\{i,j\}\subseteq \diff(r, r^j)$, and all guarantee that $v_iv_j\in E(G)$ by the fact that~\eqref{eq:vs} is a clique. Thus $v_1, \ldots, v_k$ induce a $k$-clique in $G$, which is a contradiction.
    \end{proof}

    \medskip
    For each $r\in R$ we fix an $i_r\in [k]$ satisfying Claim~3\yijia{old: Claim~1, reviewer B} and let $T_r \coloneqq \Big\{r+a\cdot e_{i_r}\in \indexset \Bigmid a\in \mathbb F^*_q\Big\}$, the same as~\eqref{eq:Tr}.

    \begin{claim}
    Every $r\in \indexset \setminus R$ can occur in at most \emph{one} $T_{r'}$ for $r'\in R$, i.e.,
    \[
        \big|\{r'\in R\mid r\in T_{r'}\}\big|\le 1.
    \]
    \end{claim}

    \begin{proof}[Proof of the claim.]
    Towards a contradiction, assume $r\in T_{r'}\cap 
    T_{r''}$ for two distinct $r', r''\in R$. Hence 
    \begin{eqnarray*}
        \diff(r',r)= \{i_{r'}\}
        & \text{and} &
        \diff(r'',r)= \{i_{r''}\}.
    \end{eqnarray*}
    Since $r'\ne r''$, we conclude
    \begin{eqnarray*}
    \diff(r',r'')= \big\{i_{r'}, i_{r''}\big\},
     & \text{and thus} &
    \dist(r', r'')\le 2.
    \end{eqnarray*}
    As $r'\in R$ and $i_{r'}\in \diff(r', r'')$, Claim~1 implies that $r''\notin R$. This is the desired contradiction.
    \end{proof}

    Finally, we conclude the proof for the soundness case of \Cref{lem:qgap}. The above two claims imply the followings.
    \begin{align*}
    R\uplus \bigcup_{r\in R} T_r\subseteq \indexset 
    & \Longrightarrow\  |R|+ \left|\bigcup_{r\in R} T_r\right| \le q^k \\
    & \Longrightarrow \ |R|+ |R|\cdot (q-1) \le q^k
        & \text{\big(by $|T_r|=q-1$ and Claim~4\big)} \\
    & \Longrightarrow |R|\cdot q \le  q^k \Longrightarrow |R|\le q^{k-1}. & \qedhere
    \end{align*}\yijia{Old: Claim~2, reviewer B}
    \end{proof}

    \medskip
    \begin{theorem}\label{thm:gapreduction}
    There is an algorithm (i.e., a reduction) $\mathbb R$ that on an input graph $G$, $k\ge 1$, and a prime power $q$, computes a graph $H= H(G,k,q)$ satisfying the following conditions.
    \begin{enumerate}
        \item[(R1)] If $G$ has a $k$-clique, then $H$ has a clique of size $q^k$.

        \item[(R2)] If $G$ does not have a $k$-clique, then $H$ has \emph{no}
        clique of size $q^{k-1}+1$.
    \end{enumerate}
    Moreover, $\mathbb R$ runs in time polynomial in $|V(G)|+ q^k$.
    \end{theorem}

    \section{Sub-polynomial Approximation Lower Bounds}
    \label{sec:lowerbound-proof}

    Equipped with Theorem~\ref{thm:gapreduction} we are ready to derive the lower bounds for the \FPT-appromxation of $k$-Clique.
    


    \medskip
    \begin{theorem}
        The $k$-clique problem has no $\FPT$-approximation with ratio $k^{o(1)}$,
        unless $\FPT= \W 1$.
    \end{theorem}

    \begin{proof}
        Towards a contradiction, assume that $\mathbb A$ is an \FPT-approximation
        for the $k$-Clique problem such that on any input graph $G$ and $k\ge 1$:
        \begin{enumerate}
            \item[(A)] If $G$ has a $k$-clique, then $\mathbb A$ outputs a clique of
            size at least $k/k^{h(k)}$, where $h: \mathbb N\to \mathbb R$ with
            \[
                \lim_{k\to \infty} h(k)= 0.
            \]
        \end{enumerate}
        We will define a function $q:\mathbb N\to \mathbb N$ such that for any $k\ge 1$, $q\coloneqq q(k)$ with $q$ being a prime, and $k'\coloneqq q^k$ with
        \begin{equation}
            \label{eq:q}
            (k')^{h(k')}< q.
        \end{equation}
        This will give us the desired contradiction, since we would have an
        $\FPT$ algorithm deciding the $k$-Clique problem: on any input $(G, k)$, we
        first compute $q\coloneqq q(k)$.
        \footnote{There is a small but annoying issue on how
        to compute $q$, but this is often ignored. For most natural functions $h$, this is easy. A more rigorous treatment requires the ``little o'' in $k^{o(1)}$ to be interpreted ``effectively'' (see~\cite{CG07} for a detailed discussion). 
        } 
        Then apply the algorithm $\mathbb R$ in Theorem~\ref{thm:gapreduction} to get a graph $H\coloneqq H(G, k, q)$. Finally
        we run the approximation algorithm $\mathbb A$ on $(H, k')$ with $k'\coloneqq q^k$.
        \begin{itemize}
            \item If $G$ has a $k$-clique, then $H$ has a $k'$-clique by (R1).
            It
            follows by (A) and~\eqref{eq:q} that the algorithm $\mathbb A$ will output a clique
            of size at least
            \[
                \frac{k'}{(k')^{h(k')}}> \frac{q^k}{q}= q^{k-1},
            \]
            i.e., at least $q^{k-1}+1$.

            \item If $G$ has no $k$-clique, then $H$ has no clique of size $q^{k-1}+1$ by (R2).
            Thus, the clique in which the algorithm $\mathbb A$ outputs must have
            size at most
            \[
                q^{k-1}.
            \]
        \end{itemize}
        Therefore, we can decide whether the original graph $G$ has a $k$-clique by
        checking whether the clique that the algorithm $\mathbb A$ computes has size at least
        $q^{k-1}+1$.

        \medskip
        It remains to show that we can define a function $q:\mathbb N\to \mathbb N$ such that for $q\coloneqq q(k)$ the inequality~\eqref{eq:q} holds.
        Let $k\in \mathbb N$.
        Since $\lim_{x\to \infty} h(x)= 0$, there is an $n_k\in \mathbb N$ such that for all $n\ge n_k$ we have
        \[
            h(n)< \frac{1}{k}.
        \]
        Define
        \[
            q = q(k) \coloneqq \min \Big\{p\in \mathbb N \Bigmid
            \text{$p$ is a prime and $p^k\ge n_k$}\Big\}.
        \]
        Clearly, $q(k)$ is well defined.
        In particular, for $k'\coloneqq q^k$ we get $k'\ge
        n_k$.
        Hence, by our choice of $n_k$
        \[
            h(k')< \frac{1}{k}.
        \]
        It follows that
        \[
            (k')^{h(k')}< (k')^{1/k}= \left(q^k\right)^{1/k}= q. 
        \]
        This is precisely~\eqref{eq:q}.
    \end{proof}

    The next lower bound was first proved in~\cite{LinRSW22-Clique}.
    \begin{theorem}[Simplified proof of \cite{LinRSW22-Clique}]
    Assuming \ETH, there is no approximation algorithm for the $k$-Clique problem with a constant ratio of running time  $t(k)n^{o(\log k)}$ for any computable function $t$.
    \end{theorem}

    \begin{proof}
    Let $c\ge 1$ and $\mathbb B$ be an algorithm that finds a clique of size $k/c$ if the input graph $G$ contains a clique of size $k$. Moreover, $\mathbb R$ runs in time $t(k)n^{o(\log k)}$, where
    $t$ is a computable function. We choose a (minimum) prime
    \[
    q\ge c+1.
    \]
 
    \medskip
    Now given a graph $G$ and $k\ge 2$. We first invoke the reduction $\mathbb R$ as stated in Theorem~\ref{thm:gapreduction} on $G$, $k$, and $q$, which produces a graph $H$
    satisfying (R1) and (R2). In particular, if $G$ has a $k$-clique, then $H$ has a clique of size
    \[
    k':= q^k.
    \]
    Otherwise, then $H$ has no clique of
    size
    \[
    q^{k-1}+1= \frac{q^k- q^{k-1}+q-1}{q-1}< \frac{q^k}{q-1}\le \frac{k'}{c},
    \]
    where the first inequality is by $k\ge 2$ and the second by $q\ge c+1$.
    Note $\mathbb R$ runs in time polynomial in
    \[
    |V(G)|+ q^k= n+ 2^{O(k)},
    \]
    which implies that
    \[
    |V(H)|\le 2^{O(k)}\poly(n).
    \]
    Next, we apply $\mathbb B$ on $H$ and $k'$. It follows that,
    \begin{itemize}
    \item {\bf Completeness.} If $G$ has a $k$-clique, then $H$ has a clique of size $k'= q^k$.
    Hence $\mathbb B$ output a clique of size at least $k'/c$.

    \item {\bf Soundness.} If $G$ does not have a $k$-clique, then $H$ has no clique of size $k'/c$.
    Thus, $\mathbb B$ cannot output a clique of size $k'/c$.
    \end{itemize}
    This means that we can decide whether the original graph $G$ has a clique of size $k$. Furthermore, observe that the running time of $\mathbb B$ is
    \[
    t(k') |V(H)|^{o(\log k')}\le t(q^k) \big(2^{O(k)}\poly(n)\big)^{o(k)}
    \le h(k) n^{o(k)}
    \]
    for some appropriate computable function $h$. This contradicts \ETH.
    \end{proof}

\section{Conclusion and Discussion}
\label{sec:conclusion}

We presented a self-reduction from $k$-Clique to $(q^k,q^{k-1})$-Gap Clique. Our reduction is simple and almost combinatorial. We simply combine the technique from the network coding theory and the use of Sidon sets. Both techniques are well-known in the literature, and no heavy machinery is involved. 
Moreover, our reduction is a self-reduction, which gives an insight into a generic gap-producing FPT-reduction for $\W 1$-hard problems.

In fact, our ideal goal is to devise a self-reduction that transforms an instance of $(k,k')$-Gap Clique into an instance of $(q^{k},q^{k'})$-Gap clique. If such a transformation exists, this will imply the following chain of reductions.

\begin{align*}
    \text{$(k,k-1)$-Gap Clique}
       \Rightarrow \text{$(q^k,q^{k-1})$-Gap Clique}
       \Rightarrow \text{$(\sigma^{q^{k}},\sigma^{q^{k-1}})$-Gap Clique}
\end{align*}

Setting $q=2$ and $\sigma=f(k)$ immediately rules out approximation ratio polynomial on $k$, say the gap of $K$ vs $K^{1/2}$. 
Since we can choose $\sigma$ to be an arbitrary function on $k$, for any computable non-decreasing function $f(K)$, we may choose $\sigma=\left(f^{-1}(K)\right)^{1/K}$ to rule out $f(K)$-approximation algorithm that runs in FPT-time. Roughly speaking, we may be just one step behind proving the total FPT-inapproximability of $k$-Clique under the $\W 1$-hardness.

\bibliographystyle{alpha}
\bibliography{ref}

\end{document}